\documentclass[runningheads]{llncs}
\usepackage{amsmath,amssymb,amsfonts}
\usepackage{xcolor}
\usepackage{thmtools}
\usepackage{thm-restate}

\makeatletter
\let\MYcaption\@makecaption
\makeatother
\usepackage{subcaption}
\usepackage{graphicx}

\allowdisplaybreaks[1]
\newcommand{\ov}[1]{\overline{#1}}

\newcommand{\eqdef}{\stackrel{\mbox{\scriptsize def}}{=}}

\newcommand{\realsnng}{\mathbb{R}^{\ge0}}
\newcommand{\cala}{\mathcal{A}}

\newcommand{\cald}{\mathcal{D}}
\newcommand{\calf}{\mathcal{F}}

\newcommand{\calo}{\mathcal{O}}

\newcommand{\calr}{\mathcal{R}}
\newcommand{\cals}{\mathcal{S}}

\newcommand{\calw}{\mathcal{W}}

\newcommand{\supp}{{\mathtt{supp}}}
\newcommand{\Dists}{\mathbb{D}} 
\newcommand{\Prob}{\mathrm{Pr}}
\newcommand{\randassign}{\ensuremath{\stackrel{\mathrm{\$}}{\leftarrow}}}

\newcommand{\diverge}[2]{\mathit{D}(#1 \parallel #2)}

\newcommand{\Dinf}[0]{\mathit{D}_{\infty}}
\newcommand{\maxdiverge}[2]{\Dinf(#1 \parallel #2)}

\newcommand{\samp}[0]{n}
\newcommand{\epsan}[0]{\varepsilon_{\!\alpha,\samp}}
\newcommand{\epsanp}[0]{\varepsilon_{\!\alpha,\samp'}}

\newcommand{\StatEL}{\mbox{\rm StatEL}}
\renewcommand{\phi} {\varphi}
\newcommand{\erightarrow}{\supset}

\newcommand{\Var}{\mathtt{Mes}}

\newcommand{\phisyn}[0]{\phi_{\mathrm{syn}}}

\newcommand{\PR}[1]{\mathop{\mathbb{P}_{#1}}}

\newcommand{\MKa}{\mathop{\mathsf{K}_{a}}}

\newcommand{\MKea}{\mathop{\mathsf{K}_{\varepsilon_{\!\alpha\!,\samp}}}}
\newcommand{\MKeap}{\mathop{\mathsf{K}_{\varepsilon_{\!\alpha\!,\samp'}}}}

\newcommand{\MKae}{\mathop{\mathsf{K}_{a,\varepsilon}}}

\newcommand{\MKaz}{\mathop{\mathsf{K}_{a,0}}}

\newcommand{\MPa}{\mathop{\mathsf{P}_{\!a}}}

\newcommand{\MPea}{\mathop{\mathsf{P}_{\!\varepsilon_{\!\alpha\!,\samp}}}}
\newcommand{\MPeapp}{\mathop{\mathsf{P}_{\!\varepsilon_{\!\alpha\!,\samp''}}}}
\newcommand{\MPe}{\mathop{\mathsf{P}_{\varepsilon}}}
\newcommand{\MPae}{\mathop{\mathsf{P}_{a,\varepsilon}}}

\newcommand{\MPaz}{\mathop{\mathsf{P}_{a,0}}}

\newcommand{\M}{\mathfrak{M}}

\newcommand{\Rea}[0]{\calr_{\varepsilon_{\!\alpha\!,n}}}
\newcommand{\Reps}[0]{\calr_{\!\varepsilon}}
\newcommand{\Raeps}[0]{\calr_{a,\varepsilon}}
\newcommand{\Raz}[0]{\calr_{a,0}}

\newcommand{\wre}[0]{\mathit{w_{\sf real}}}
\newcommand{\wid}[0]{\mathit{w_{\sf ideal}}}

\newcommand{\alg}{\mathit{A}}

\newif\ifcommentson\commentsonfalse

\newif\ifconferenceon\conferenceonfalse
\ifconferenceon
\newcommand{\arxiv}[1]{}
\newcommand{\conference}[1]{#1}
\newcommand{\conferenceShort}[1]{}
\else
\newcommand{\arxiv}[1]{#1}
\newcommand{\conference}[1]{}
\newcommand{\conferenceShort}[1]{}
\fi
\newcommand{\journal}[1]{}

\ifcommentson
\newcommand{\commentsize}[0]{.90\textwidth}
\newcommand{\commentYK}[1]{\begin{center} \parbox{\commentsize}{\textbf{\textcolor{black}{Comment Y.}} \textcolor{red}{#1} }\end{center}}
\newcommand{\replyYK}[1]{\begin{center} \parbox{\commentsize}{\textbf{Reply Y.} \textcolor{blue}{#1} }\end{center}}
\newcommand{\commentY}[1]{\marginpar{\footnotesize \color{red} {\bf Y:} \textsf{\scriptsize #1}}}
\newcommand{\replyY}[1]{\marginpar{\footnotesize \color{red} {\bf Y:} \textsf{\scriptsize #1}}}
\else
\newcommand{\commentYK}[1]{}
\newcommand{\replyYK}[1]{}
\newcommand{\commentY}[1]{}
\newcommand{\replyY}[1]{}
\fi

\newcommand{\colorR}[1]{\textcolor{red}{#1}}

\newcommand{\pagelimitmarker}[1]{~\\ {\colorR{\ifthenelse{\thepage>#1}{\Huge Exceeding the page limit}{\huge Within the page limit}}}~\\ {\huge{\colorR{~~Page Limit\,\,\,\,\, = #1}}}~\\ {\huge{\colorR{~~Current Page = $\thepage$}}}}

\begin{document}
\title{Statistical Epistemic Logic
\thanks{This work was supported by JSPS KAKENHI Grant Number JP17K12667, and by Inria under the project LOGIS.}}
\author{Yusuke Kawamoto\inst{1}\orcidID{0000-0002-2151-9560}}
\authorrunning{Y. Kawamoto}
\institute{National Institute of Advanced Industrial Science \\ 
and Technology (AIST), Tsukuba, Japan \\
\conference{\email{yusuke.kawamoto.aist@gmail.com}}}
\maketitle              
\begin{abstract}
We introduce a modal logic for describing statistical knowledge, which we call \emph{statistical epistemic logic}. We propose a Kripke model dealing with probability distributions and stochastic assignments, and show a stochastic semantics for the logic. To our knowledge, this is the first semantics for modal logic that can express the statistical knowledge dependent on non-deterministic inputs and the statistical significance of observed results. By using statistical epistemic logic, we express a notion of statistical secrecy with a confidence level. We also show that this logic is useful to formalize statistical hypothesis testing and differential privacy in a simple and abstract manner. 
\end{abstract}

\section{Introduction}
\label{sec:intro}
Knowledge representation and reasoning have been studied in two research areas: \emph{logic} and \emph{statistics}.
Broadly speaking, logic describes our knowledge using formal languages and reasons about it using symbolic techniques,
while statistics interprets collected data having random variation and infers properties of their underlying probability models.
As research advances demonstrate, logical and statistical approaches are respectively successful in many applications, including 
artificial intelligence, 
software engineering, and information security.

The techniques of these two approaches are basically orthogonal and could be integrated to get the best of both worlds.
For example, in a large system with artificial intelligence (e.g., an autonomous car), 
both rule-based knowledge and statistical machine learning models may be used, and the way of combining them 
would be crucial to the performance and security of the whole system.
However, even in theoretical research on knowledge models, there still remains much to be done to integrate techniques from the two approaches.
For a very basic example, 
\emph{epistemic logic}~\cite{vonWright:51:book}, a formal logic for representing and reasoning about knowledge, has not yet been able to model ``statistical knowledge'' 
with sampling and statistical significance, although a lot of epistemic models~\cite{Fagin:95:book,Halpern:03:book,Halpern:03:CSFW}
have been proposed so far.

One of the important challenges in integrating logical and statistical knowledge is to design a logical model for statistical knowledge, 
which can be updated by a limited number of sampling of probabilistic events and by the non-deterministic inputs from an external environment.
Here we note that non-deterministic inputs are essential to model the security of the system,
because we usually do not have a prior knowledge of the probability distribution of adversarial inputs and need to reason about the worst scenarios caused by the attack.
Nevertheless, to the best of our knowledge, no previous work on epistemic logic has proposed an abstract model for the statistical knowledge that involves non-deterministic inputs and the statistical significance of observed results.

In the present paper, we propose an epistemic logic for describing statistical knowledge.
To define its semantics, we introduce a variant of a Kripke model~\cite{Kripke:63:MLQ} in which each possible world is defined as a probability distribution of states and each variable is probabilistically assigned a value.
In this model, the stochastic behaviour of a system is modeled as a distribution of states at each world, and each non-deterministic input to the system corresponds to a distinct possible world.
As for applications of this model, we define an accessibility relation as a statistical distance between distributions of observations, and show that our logic is useful to formalize statistical hypothesis testing and differential privacy~\cite{Dwork:06:ICALP} of statistical data.

\paragraph{Our contributions.}
The main contributions of this work are as follows:
\begin{itemize}
\item We introduce a modal logic, called \emph{statistical epistemic logic} (\StatEL{}), to describe statistical knowledge.
\item We propose a Kripke model incorporating probability distributions and stochastic assignments by regarding each possible world as a distribution of states and by defining an accessibility relation using a metric/divergence between distributions.
\item We introduce a stochastic semantics for \StatEL{} based on the above models.
As far as we know, this is the first semantics for modal logic that can express the statistical knowledge dependent on non-deterministic inputs and the statistical significance of observed results.
\item We present basic properties of the probability quantification and epistemic modality in \StatEL{}.
In particular, we show that the transitivity and Euclidean axioms rely on the agent's capability of observation in our model.
\item By using \StatEL{} we introduce a notion of statistical secrecy with a significance level $\alpha$.
We also show that \StatEL{} is useful to formalize statistical hypothesis testing 
and differential privacy in a simple and abstract manner.
\end{itemize}

\paragraph*{Paper organization.}
The rest of this paper is organized as follows.
Section~\ref{sec:preliminaries} introduces background and notations used in this paper.
Section~\ref{sec:running-example} presents an example of coin flipping to explain the motivation for a logic of statistical knowledge.
Section~\ref{sec:logic} shows the syntax and semantics of the statistical epistemic logic \StatEL{}.
Section~\ref{sec:properties} presents basic properties of the logic.
As for applications, Sections~\ref{sec:formal:HT} and~\ref{sec:formal:DP} respectively model statistical hypothesis testing and statistical data privacy using \StatEL{}.
Section~\ref{sec:related} presents related work and Section~\ref{sec:conclude} concludes.

\section{Preliminaries}
\label{sec:preliminaries}
In this section we 
recall the definitions of divergence and metrics, which are used 
in later sections 
to quantitatively model an agent's capability of distinguishing possible worlds.

\subsection{Notations}
\label{subsec:notations-pd}

Let $\realsnng$ be the set of non-negative real numbers,
and $[0, 1] = \{ r\in\realsnng \mid r \le 1 \}$.
We denote by $\Dists\calo$ the set of all probability distributions over a set~$\calo$. 
For a finite set $\calo$ and a distribution $\mu\in\Dists\calo$, the probability of sampling a value $y$ from $\mu$ is denoted by $\mu[y]$.
For a subset $R\subseteq\calo$, let $\mu[R] = \sum_{y\in R} \mu[y]$.
The \emph{support} of a distribution $\mu$ over a finite set $\calo$ is
$\supp(\mu) = \{ v \in \calo \colon \mu[v] > 0 \}$.
For a set $\cald$, a randomized algorithm $\alg:\cald\rightarrow\Dists\calo$ and a set $R\subseteq\calo$ we denote by $\alg(d)[R]$ the probability that given input $d\in\cald$, $\alg$ outputs one of the elements of $R$.

\subsection{Metric and Divergence}
\label{subsec:divergence}

A \emph{metric} over a non-empty set $\calo$ is a function $d: \calo\times\calo \rightarrow \realsnng$ such that for all $y, y', y'' \in \calo$,
(i)
$d(y, y') \ge 0$;
(ii)
$d(y, y') = 0$ iff $y = y'$;
(iii)
$d(y, y') = d(y', y)$;
(iv)
$d(y, y'') \le d(y, y') + d(y', y'')$.
Recall that (iii) and (iv) are respectively referred to as symmetry and subadditivity.

A \emph{divergence} over a non-empty set $\calo$ is a function $\diverge{\cdot}{\cdot}: \Dists\calo\times\Dists\calo\rightarrow\realsnng$ such that 
for all $\mu,\mu'\in\Dists\calo$, 
(i)
$\diverge{\mu}{\mu'}\ge0$ and 
(ii)
$\diverge{\mu}{\mu'}=0$ iff $\mu=\mu'$.
Note that a divergence may not be symmetric or subadditive.

To describe a statistical hypothesis testing in Section~\ref{sec:formal:HT},
we recall the definition of $\chi^2$ divergence due to Pearson~\cite{Pearson:1900:LEDPMJS} as follows:

\begin{definition}[Pearson's $\chi^2$ divergence]\label{def:chi-divergence}\rm
Given two distributions $\mu, \mu'$ over a finite set $\calo$,\,
the \emph{$\chi^2$-divergence} $D_{\chi^2}(\mu \parallel  \mu')$ of $\mu$ from $\mu'$ is defined by:
\[
D_{\chi^2}(\mu \parallel  \mu') =
\!\sum_{y\in\supp(\mu)}\hspace{-1ex} 
\frac{ (\mu'[y] - \mu[y])^2 }{ \mu[y] }
{.}
\]
\emph{$\chi^2$ statistics} is the multiplication of $\chi^2$-divergence with a sample size~$\samp$.
\end{definition}

To introduce a notion of statistical data privacy in Section~\ref{sec:formal:DP}, 
we recall the definition of 
the max-divergence $\Dinf$ 
as follows.
\begin{definition}[Max divergence]\label{def:max-divergence}\rm
For two distributions $\mu, \mu'$ over a finite set~$\calo$,\, 
the \emph{max divergence} $\maxdiverge{\mu}{\mu'}$ 
of $\mu$ from $\mu'$ is defined~by:
\[
\maxdiverge{\mu}{\mu'} = 
\max_{\substack{R\subseteq\supp(\mu)}}\, 
\ln \frac{ \mu[R] }{ \mu'[R] }
{.}
\]
\end{definition}

Note that neither $D_{\chi^2}$ nor 
$\Dinf$ 
is symmetric.

\section{Motivating Example}
\label{sec:running-example}
In this section we present a motivating example to explain why we need to introduce a new model for epistemic logic to describe statistical knowledge.

\begin{figure}[t]
\begin{subfigure}[t]{0.45\textwidth}
\centering
\mbox{\raisebox{0pt}{\includegraphics[ width=0.99\textwidth]{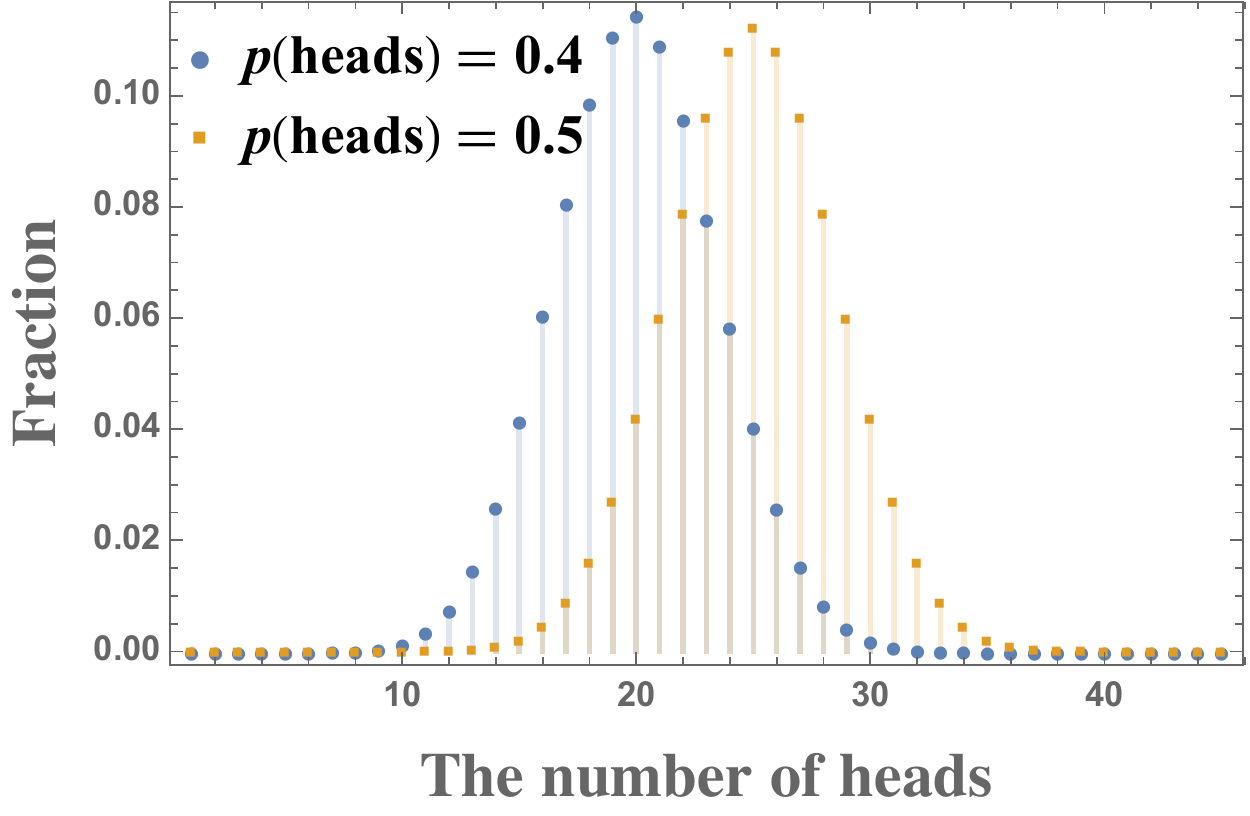}}}
\caption{Given $50$ coin flips, the two distributions overlap much.
\label{fig:binomial50}}
\end{subfigure}
\hfill
\begin{subfigure}[t]{0.45\textwidth}
\centering
 \mbox{\raisebox{0pt}{\includegraphics[ width=1.00\textwidth]{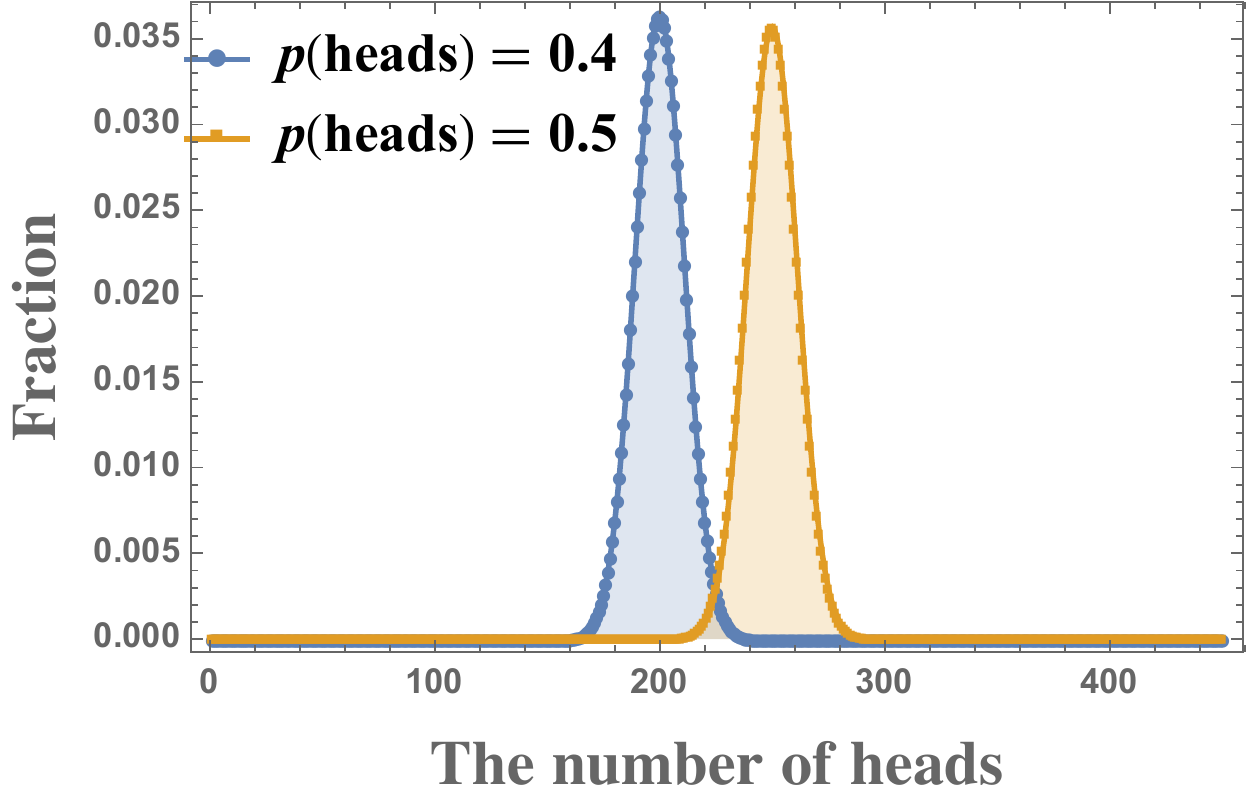}}}
\caption{Given $500$ coin flips, the two distributions are distinguished more clearly.
\label{fig:binomial500}}
\end{subfigure}
\vspace{-2ex}
\caption{The frequency distributions of the numbers of heads in coin flipping.
\label{fig:frequency:heads}}
\end{figure}

\begin{example}[Coin flipping]\label{eg:coin}
Let us consider a simple running example of flipping a coin in two possible worlds $w_0$ and $w_1$ respectively.
We assume that in the world $w_0$ the coin is fair 
(represented by $p(heads) = 0.5$),
whereas in $w_1$ the probability of getting a heads is $0.4$
(represented by $p(heads) = 0.4$).
Here we do not have any prior belief on the probabilities of the worlds $w_0$ and $w_1$. 
This does not mean $p(w_0) = p(w_1) = 0.5$, but means 
we have no idea on the values of $p(w_0)$ and $p(w_1)$ at all,
i.e., either $w_0$ or $w_1$ is chosen non-deterministically.
\end{example}

When we flip a coin just once and observe its outcome (heads or tails), we do not know whether the coin is fair or biased, that is, we cannot tell whether we are located in the world $w_0$ or $w_1$.

As shown in Fig.~\ref{fig:frequency:heads}, however, when we increase the number $\samp$ of coin flips, we can more clearly see the difference between the numbers of getting heads in $w_0$ and in $w_1$.
If the fraction of observing heads goes to $0.5$ (resp. $0.4$), then we learn we are located in the world $w_0$ (resp. $w_1$) with a stronger confidence,
namely, we have a stronger belief that the coin is fair (resp. biased).
This implies that a larger number of observing the outcome enables us to distinguish two possible worlds more clearly, hence to obtain a stronger belief.

\begin{figure}[t]\label{fig:compositions}
\centering
\begin{subfigure}[t]{0.47\textwidth}
\centering
\begin{picture}(140, 50)
 \put(70,25){\oval(140,50)}
 \put(35,30){\oval(40,20)}
 \put(105,30){\oval(40,20)}
 \put( 24,  28){$heads$}
 \put( 96,  28){$tails$}
 \put( 10,  10){\scriptsize probability 0.5}
 \put( 79,  10){\scriptsize probability 0.5}
\end{picture}
\caption{The world $w_0$ with the fair coin.
\label{fig:world:fair}}
\end{subfigure}\hspace{0.1ex}\hfill
\begin{subfigure}[t]{0.47\textwidth}
\centering
\begin{picture}(140, 50)
 \put(70,25){\oval(140,50)}
 \put(35,30){\oval(40,20)}
 \put(105,30){\oval(40,20)}
 \put( 24,  28){$heads$}
 \put( 96,  28){$tails$}
 \put( 10, 10){\scriptsize probability 0.4}
 \put( 79, 10){\scriptsize probability 0.6}
\end{picture}
\caption{The world $w_1$ with the biased coin.
\label{fig:world:biased}}
\end{subfigure}
\caption{One of the possible worlds (i.e., $w_0$ or $w_1$) is chosen non-deterministically. 
Then one of the states (i.e., $heads$ or $tails$) is chosen probabilistically.
\label{fig:worlds}}
\end{figure}
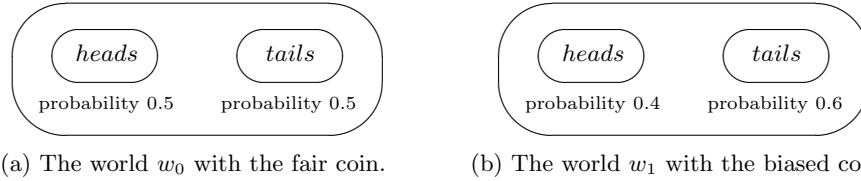

To model such statistical beliefs, we regard each possible world as a probability distribution of two states $heads$ and $tails$
as shown in Fig.~\ref{fig:worlds}
(e.g., $w_1[heads] = 0.4$ and $w_1[tails] = 0.6$).
Then for a divergence $D$ between two distributions, we define an accessibility relation $\Reps$ between worlds such that 
for any worlds $w$ and $w'$,\,
$(w, w')\in\Reps$ iff $\diverge{w\!}{\!w'} \le \varepsilon$.
Then $(w_0, w_1)\in\Reps$ for a smaller threshold $\varepsilon$ represents that a larger number of sampling is required to distinguish $w_0$ from~$w_1$.

This relation $\Reps$ is used to formalize statistical knowledge in a model of epistemic logic in Section~\ref{sec:logic}.
Intuitively, given a threshold $\varepsilon$ determining a confidence level,
we say that we know a proposition $\phi$ in a world $w$ if $\phi$ is satisfied in all possible worlds that are indistinguishable from $w$ in terms of $\Reps$.
In Section~\ref{sec:formal:HT} we will revisit the coin flipping example to see how we formalize it using our logic.

To our knowledge, no previous work on epistemic logic has modeled a statistical knowledge that depends on the agent's capability of observing events.
In fact, in most of the Kripke models used in previous work, a possible world represents a single state instead of a probability distribution of states, hence the relation between possible worlds does not involve the probability of distinguishing them.
Therefore, no prior work on epistemic logic has proposed an abstract model for the statistical knowledge that involves the sample size of observing random variables and the statistical significance of the observed results.

\section{Statistical Epistemic Logic (\StatEL{})}
\label{sec:logic}
In this section we introduce the syntax and semantics of the \emph{statistical epistemic logic} (\StatEL{}).

\subsection{Syntax}
\label{sub:syntax}

We first present the syntax of the statistical epistemic logic as follows.
To express both deterministic and probabilistic properties, we introduce two levels of formulas: \emph{static formulas} and \emph{epistemic formulas}.
Intuitively, a static formula represents a proposition that can be satisfied at a state with probability $1$, while an epistemic formula represents a proposition that can be satisfied at a probability distribution of states with some probability.

Formally, let $\Var$ be a set of symbols called \emph{measurement variables}, and
$\Gamma$ be a set of atomic formulas of the form $\gamma(x_1, x_2, \ldots, x_n)$ for a predicate symbol $\gamma$ and $x_1, x_2, \ldots, x_n\in\Var$ ($n \ge 0$).
Let $I \subseteq [0, 1]$ be a finite union of intervals, and $\cala$ be a finite set of indices (typically associated with the names of agents and/or statistical tests).
Then the static and epistemic formulas are defined by:
\begin{itemize}
\item[] Static formulas:~~
$\psi \mathbin{::=}
 \gamma(x_1, x_2, \ldots, x_n) \mid
 \neg \psi \mid \psi \wedge \psi$
\item[] Epistemic formulas:~~
$\phi \mathbin{::=}
 \PR{I} \psi \mid \neg \phi \mid \phi \wedge \phi \mid
 \psi \erightarrow \phi \mid \MKa \phi$
\end{itemize}
where 
$a\in\cala$.
Let $\calf$ be the set of all epistemic formulas.
Note that we have no quantifiers over measurement variables. (See Section~\ref{sub:interpretation}.)

The \emph{probability quantification} $\PR{I} \psi$ represents that a static formula $\psi$ is satisfied with a probability belonging to a set $I$.
For instance, $\PR{(0.5, 1]} \psi$ represents that $\psi$ holds with a probability greater than $0.5$.
The \emph{non-classical implication} $\erightarrow$ is used to represent conditional probabilities.
For example, by $\psi_0 \erightarrow \PR{I} \psi_1$ we represent that the conditional probability of $\psi_1$ given $\psi_0$ is included in a set $I$.
The \emph{epistemic knowledge} $\MKa \phi$ expresses that an agent $a$ knows $\phi$.
The formal meaning of these operators will be shown in the definition of semantics.

As syntax sugar, we use \emph{disjunction} $\vee$, \emph{classical implication} $\rightarrow$, and \emph{epistemic possibility operator} $\MPa$, defined by:
$\phi_0 \vee \phi_1 \mathbin{::=} \neg (\neg \phi_0 \wedge \neg \phi_1)$,
$\phi_0 \rightarrow \phi_1 \mathbin{::=} \neg \phi_0 \vee \phi_1$,
and $\MPa{\phi} \mathbin{::=} \neg \MKa \neg \phi$.
When $I$ is a singleton $\{ i \}$, we abbreviate $\PR{[i, i]}$ as $\PR{i}$.

\subsection{Modeling of Systems}
\label{sub:system}

In this work we deal with a simple stochastic system with measurement variables.
Let $\calo$ be the finite set of all data that can be assigned to the measurement variables in $\Var$.
We assume that all possible worlds share the same domain $\calo$.
We define a \emph{stochastic system} as a pair $(S, \sigma)$ consisting of:
\begin{itemize}
\item a stochastic program $S$ that deals with input and output data through measurement variables in $\Var$, behaves deterministically or probabilistically (by using some randomly generated data), and terminates with probability $1$;
\item a \emph{stochastic assignment} $\sigma:\Var\rightarrow\Dists\calo$ representing that each measurement variable $x$ has an observed value $v$ with probability $\sigma(x)[v]$.
\end{itemize}
Here we present only a general model and do not specify the data type of those measurement variables,
which can be (sequences of) bit strings, floating point numbers, texts, or other types of data.
Thanks to the assumption on the program termination and on the finite range of data, the program $S$ can reach finitely many states.
For the sake of simplicity, our model does not take timing into account.
Extension to time and temporal modality is left for future work.

\subsection{Distributional Kripke Model}
\label{sub:Kripke}

To define a semantics for \StatEL{}, we recall the notion of a Kripke model~\cite{Kripke:63:MLQ}:
\begin{definition}[Kripke model] \label{def:Kripke-model} \rm
Given a set $\Gamma$ of atomic formulas, a \emph{Kripke model} is defined as a triple $(\calw, \calr, V)$ consisting of a non-empty set $\calw$, a binary relation $\calr$ on $\calw$, and a function 
$V$ that maps each atomic formula $\gamma\in\Gamma$ to a subset $V(\gamma)$ of $\calw$.
The set $\calw$ is called a \emph{universe}, its elements are called \emph{possible worlds}, $\calr$ is called an \emph{accessibility relation}, and $V$ is called a \emph{valuation}.
\end{definition}

Now we introduce a Kripke model called a ``distributional'' Kripke model where each possible world is a probability distribution of states over $\cals$ and 
each world $w$ is associated with a stochastic assignment $\sigma_w$ to measurement variables.

\begin{definition}[Distributional Kripke model] \label{def:dist-Kripke-model} \rm
Let $\cala$ be a finite set of indices (typically associated with the names of agents and/or statistical tests),
$\cals$ be a finite set of states\footnote{It is left for future work to investigate the case of infinite numbers of states.}, and $\calo$ be a finite set of data.
A \emph{distributional Kripke model} 
is a tuple 
$\M =(\calw, (\calr_a)_{a\in\cala}, (V_s)_{s\in\cals})$ 
consisting of:
\begin{itemize}
\item a non-empty set\footnote{Since $\calw$ is not a multiset, each world in $\calw$ is a different distribution of states. However, this is still expressive enough when we take $\cals$ to be sufficiently large.
} $\calw$ of probability distributions of states over $\cals$;
\item for each $a\in\cala$, an accessibility relation $\calr_a \subseteq \calw \times \calw$;
\item for each $s\in\cals$, a valuation $V_s$ that maps each $k$-ary predicate $\gamma$ to a set $V_s(\gamma) \subseteq \calo^k$.
\end{itemize}
We assume that each $w\in\calw$ is associated with a function $\rho_w: \Var\times\cals\rightarrow\calo$ that maps each measurement variable $x$ to its value $\rho_w(x, s)$ observed at a state~$s$.
We also assume that each state $s$ in a world $w$ is associated with the assignment $\sigma_s: \Var\rightarrow\calo$ defined by $\sigma_s(x) = \rho_w(x, s)$.
\end{definition}

Note that this model assumes a constant domain $\calo$; i.e., all measurement variables range over the same set $\calo$ in every world.
Since each world $w$ is a probability distribution of states, we denote by $w[s]$ the probability that a state $s$ is sampled from $w$.
Then the probability that a variable $x$ has a value $v$ in a world $w$ is given by:
\[
\sigma_w(x)[v] = 
\sum_{\substack{s\in\supp(w),\, \sigma_s(x) = v}} w[s]
{.}
\]
This means that when a state $s$ is drawn from the distribution $w$, an input value $\sigma_s(x)$ is sampled from the distribution $\sigma_w(x)$.

\subsection{Divergence-based Accessibility Relation}
\label{sub:divergence-relation}

Next we introduce a family of accessibility relations used in typical statistical inferences.
Since many notions of statistical distance are not metrics but divergences, we introduce an accessibility relation based on a divergence as follows.

Suppose that an agent $a$ observes some data through a single measurement variable $x$.
Then the distribution of the observed data at a world $w$ is represented by $\sigma_{w}(x)$. 
Assume that the agent $a$ distinguishes distributions in terms of a divergence $\diverge{\cdot\!}{\!\cdot}: \Dists\calo \times \Dists\calo \rightarrow \realsnng$.
Then given a threshold $\varepsilon \ge 0$, we define a \emph{divergence-based accessibility relation} $\Raeps$ by:
\[
\Raeps \eqdef
\left\{ (w, w') \in \calw\times\calw \mid D(\sigma_{w}(x) \parallel \sigma_{w'}(x)) \le \varepsilon
\right\}
{.}
\]
For a smaller value of $\varepsilon$, the capability of distinguishing worlds is stronger.

If $D$ is a metric instead,
we call $\Raeps$ a \emph{metric-based accessibility relation}.
We often omit $a$ to write $\Reps$ when we do not compare different agents' knowledge.

Intuitively, $(w, w')\in\Raeps$ represents that the distribution of the data observed in $w$ is indistinguishable from that in $w'$ in terms of $D$.
By the definition of a divergence/metric $D$,  $D(\sigma_{w}(x) \parallel \sigma_{w'}(x)) = 0$ implies $\sigma_{w}(x) = \sigma_{w'}(x)$.
Therefore, the relation $\calr_{a,0}$ expresses that the agent $a$ has an unlimited capability of observing the distributions $\sigma_{w}(x)$ and $\sigma_{w'}(x)$.
In Sections~\ref{sec:formal:HT} and~\ref{sec:formal:DP} we will show examples of divergence-based accessibility relations.

\subsection{Stochastic Semantics}
\label{sub:interpretation}

In this section we define the \emph{stochastic semantics} for the \StatEL{} formulas over a distributional Kripke model 
$\M =(\calw, (\calr_a)_{a\in\cala}, (V_s)_{s\in\cals})$ 
with $\calw = \Dists\cals$.

The interpretation of static formulas $\psi$ at a state $s$ is given by:
\begin{align*}
s \models \gamma(x_1, x_2, \ldots, x_k) 
& ~\mbox{ iff }~
(\sigma_s(x_1), \sigma_s(x_2), \ldots, \sigma_s(x_k)) \in V_s(\gamma)
\\
s \models \neg \psi 
& ~\mbox{ iff }~
s \not\models \psi 
\\
s \models \psi \wedge \psi'
& ~\mbox{ iff }~
s \models \psi
~\mbox{ and }~
s \models \psi'
{.}
\end{align*}
Note that the satisfaction of the static formulas does not involve probability.

To interpret the non-classical implication $\erightarrow$, we define the \emph{restriction} $w|_\psi$ of a world $w$ to a state formula $\psi$ as follows.
If there exists a state $s$ such that $w[s] > 0$ and $s \models \psi$, then
$w|_\psi$ can be defined as the distribution over the finite set $\cals$ of states such that:
\begin{align*}
w|_\psi[s] = 
\begin{cases}
\frac{w[s]}{\sum_{s': s' \models \psi} w[s']}
& \mbox{ if $s \models \psi$}
\\
0 & \mbox{ otherwise.}
\end{cases}
\end{align*}
Then $\sum_{s} w|_\psi[s] = 1$.
Note that $w|_\psi$ is undefined if $w$ does not have a state $s$ that satisfies $\psi$ and has a non-zero probability in $w$.

Now we define the interpretation of epistemic formulas at a world $w$ in $\M$~by:
\begin{align*}
\M, w \models \PR{I} \psi
& ~\mbox{ iff }~
\Pr\!\left[ s \randassign w :~ s \models \psi \right] \in I
\\
\M, w \models \neg \phi
& ~\mbox{ iff }~
\M, w \not\models \phi
\\
\M, w \models \phi \wedge \phi'
& ~\mbox{ iff }~
\M, w \models \phi
~\mbox{ and }~
\M, w \models \phi'
\\
\M, w \models \psi \erightarrow \phi
& ~\mbox{ iff }~
\mbox{$w|_{\psi}$ is defined and }~
\M, w|_{\psi} \models \phi
\\
\M, w \models \MKa \phi
& ~\mbox{ iff }~
\mbox{for every $w'$ s.t. $(w, w') \in \calr_a$, }~
\M, w' \models \phi
{,}
\end{align*}
where $s \randassign w$ represents that a state $s$ is sampled from the distribution $w$.

Finally, the interpretation of an epistemic formula $\phi$ in $\M$ is given by:
\begin{align*}
\M \models \phi
& ~\mbox{ iff }~
\mbox{for every world $w$ in $\M$, }~
\M, w \models \phi
{.}
\end{align*}
We remark that in each world $w$, measurement variables can be interpreted using $\sigma_w$,  as shown in Section~\ref{sub:Kripke}. 
This allows one to assign different values to distinct occurrences of a variable in a formula;
E.g., in $\phi(x) \rightarrow \MKa \phi'(x)$,\, the measurement variable $x$ occurring in $\phi(x)$ can be interpreted using $\sigma_{w}$ in a world $w$, while $x$ in $\phi'(x)$ can be interpreted using $\sigma_{w'}$ in another $w'$ s.t. $(w, w')\in \calr_a$.

Note that our semantics for probability quantification is different from that in the previous work.
Halpern~\cite{Halpern:90:AI} shows two approaches to defining semantics: giving probabilities (1) on the domain and (2) on possible worlds.
However, our semantics is different from both. 
It defines probabilities on the states belonging to a possible world, while each world is not assigned a probability.
Hence, unlike Halpern's approaches, our model can deal with both probabilistic behaviours of systems and non-deterministic inputs from an external environment.

We also remark that
\StatEL{} can be used to formalize conditional probabilities.
If the conditional probability of satisfying a static formula $\psi_1$ given another static formula $\psi_0$ is included in a set $I$ at a world $w$, then we have 
$\Pr\!\left[ s \randassign w|_{\psi_0}: s \models \psi_1 \right]\!\in~I$,
hence we obtain
$\M, w \models \psi_0 \erightarrow \PR{I} \psi_1$.

\section{Basic Properties of \StatEL{}}
\label{sec:properties}
In this section we present basic properties of \StatEL{}.
In particular, we show the transitivity and Euclidean axioms rely on the agent's capability of observation.

\subsection{Properties of Probability Quantification}
\label{sub:property:PIQ}

We can define a dual operator of $\PR{I}$ as follows.
Given a finite union $I \subseteq [0, 1]$ of disjoint intervals,
let $I^c \eqdef [0,1] \setminus I$ and $\ov{I}\,\, \eqdef \{ 1 - p \mid p \in I \}$.
Then $\ov{I^c} = \ov{I}^c$.
Negation with $\PR{I}$ has the following properties.

\begin{restatable}[Negation with probability quantification]{prop}{AxiomsNotPIQ}
\label{prop:not-PIQ}
For any world $w$ in a model $\M$ and any static formula $\psi$, we have:
\begin{enumerate}
\item 
$\M, w \models \neg \PR{I} \psi$ ~iff~ 
$\M, w \models \PR{I^c} \psi$
\item 
$\M, w \models \PR{I} \neg \psi$ ~iff~
$\M, w \models \PR{\ov{I}} \psi$.
\end{enumerate}
\end{restatable}

By Proposition~\ref{prop:not-PIQ}, $\neg \PR{I} \neg \psi$ is logically equivalent to $\PR{\ov{I^c}} \psi$.
For instance, $\neg \PR{(0,1]} \neg \psi$ is equivalent to $\PR{1} \psi$, and 
$\neg \PR{[0,1)} \neg \psi$ is equivalent to $\PR{0} \psi$.

\subsection{Properties of Epistemic Modality}
\label{sub:property:epistemic}
Next we show some properties of epistemic modality.
As with the standard modal logic, \StatEL{} satisfies the necessitation rule and distribution axiom.

\begin{restatable}[Minimal properties]{prop}{AxiomsMinimal}
\label{prop:axioms:basic}
For any distributional Kripke model $\M$, any $\phi, \phi_0, \phi_1\in\calf$, and any $a\in\cala$, we have:
\begin{itemize}
\item[]{\rm ({\bf N})} necessitation:~ 
$\M \models \phi$ implies $\M \models \MKa \phi$
\item[]{\rm ({\bf K})} distribution:~ 
$\M \models \MKa (\phi_0 \rightarrow \phi_1) \rightarrow  (\MKa \phi_0 \rightarrow \MKa \phi_1)$.
\end{itemize}
\end{restatable}

The satisfaction of other properties depends on the definition of the accessibility relation.
Since many notions of statistical distance are not metrics but divergences, we present some basic properties when $\M$ has a divergence-based accessibility relation:
$\Raeps = 
\left\{ (w, w') \in \calw\times\calw \mid D(\sigma_{w}(x) \parallel \sigma_{w'}(x)) \le \varepsilon
\right\}$.

\begin{restatable}[Properties with divergence-based accessibility]{prop}{AxiomsDivergence}
\label{prop:axioms:divergence}
Let $a\in\cala$ and $\varepsilon \ge \varepsilon' \ge 0$.
For any distributional Kripke model $\M$ with a divergence-based accessibility relation $\Raeps$ and any $\phi\in\calf$, we have:
\begin{itemize}
\item[]{\rm ({\bf T})} reflexivity:~ 
$\M \models \MKae \phi \rightarrow \phi$
\item[]{\rm ($\bf{\ge}$)} comparison of observability:~ 
$\M \models {\mathop{\mathsf{K}_{a,\varepsilon}}} \phi \rightarrow {\mathop{\mathsf{K}_{a,\varepsilon'}}} \phi$.
\end{itemize}
If $\Raeps$ is symmetric (e.g., based on the Jensen-Shannon divergence~\cite{Lin:91:TIT}) then:
\begin{itemize}
\item[]{\rm ({\bf B})} symmetry:~ 
$\M \models \phi \rightarrow \MKae \MPae \phi$.
\end{itemize}
\end{restatable}
Here the axiom $({\bf{\ge}})$ represents that 
an agent having a stronger capability of distinguishing worlds may have more beliefs.

Finally, we show some properties when $\Raeps$ is based on a metric (e.g. the $p$-Wasserstein metric~\cite{Vaserstein:69:PPI}, including the Earth mover's distance).

\begin{restatable}[Properties with metric-based accessibility]{prop}{AxiomsMetric}
\label{prop:axioms:metric}
Let $a\in\cala$ and $\varepsilon, \varepsilon' \ge 0$.
For any distributional Kripke model $\M$ with a metric-based accessibility relation $\Raeps$ and any $\phi\in\calf$,
we have {\rm ({\bf T})}reflexivity, {\rm ({\bf B})}symmetry,~and:
\begin{itemize}
\item[]{\rm ({\bf 4q})} quantitative transitivity:~ 
$\M \models {\mathop{\mathsf{K}_{a,\varepsilon+\varepsilon'}}} \phi \rightarrow \MKae {\mathop{\mathsf{K}_{a,\varepsilon'}}} \phi$
\item[]{\rm ({\bf 5q})} relaxed Euclidean:~ 
$\M \models {\mathop{\mathsf{P}_{a,\varepsilon}}} \phi \rightarrow {\mathop{\mathsf{K}_{a,\varepsilon'}}} {\mathop{\mathsf{P}_{a,\varepsilon+\varepsilon'}}} \phi$.
\end{itemize}
If the agent has an unlimited capability of observation (i.e., $\varepsilon = \varepsilon' = 0$), then:
\begin{itemize}
\item[]{\rm ({\bf 4})} transitivity:~ 
$\M \models \MKaz \phi \rightarrow \MKaz \MKaz \phi$
\item[]{\rm ({\bf 5})} Euclidean:~ 
$\M \models \MPaz \phi \rightarrow \MKaz \MPaz \phi$.
\end{itemize}
\end{restatable}

By this proposition, for $\varepsilon = 0$, \StatEL{} has the axioms of {\bf S5}, hence 
the epistemic operator $\MKaz$ represents knowledge rather than beleif.

However, if the agent has a limited observability (i.e., $\varepsilon>0$), then neither transitivity nor Euclidean may hold.
This means that, even when he know whether $\phi$ holds or not with some confidence, he may not be perfectly confident that he knows it.

\section{Modeling Statistical Hypothesis Testing Using \StatEL{}}
\label{sec:formal:HT}
In this section we formalize statistical hypothesis testing by using \StatEL{} formulas, and introduce a notion of statistical secrecy with a confidence level.

\subsection{Statistical Hypothesis Testing}
\label{sub:HT}
A \emph{statistical hypothesis testing} is a method of statistical inference to check whether given datasets provide sufficient evidence to support some hypothesis.
Typically, given two datasets, a \emph{null hypothesis} $H_0$ is defined to claim that there is no statistical relationship between the two datasets (e.g., no difference between the result of a medical treatment and the placebo effect), while an \emph{alternative hypothesis} $H_1$ represents that there is some relationship between them (e.g., the result of a medical treatment is better than the placebo effect).

Before performing a hypothesis test, we specify a \emph{significance level} \journal{(\emph{type I error rate})} $\alpha$, i.e., the probability that the test might reject the null hypothesis $H_0$, given that $H_0$ is true.
Typically, $\alpha$ is $0.05$ or $0.01$.
$1 - \alpha$ is called a \emph{confidence level}.
\journal{Symmetrically, a \emph{type II error rate} $\beta$ is the probability that the test fails to reject $H_0$, given that $H_0$ is false.
$1 - \beta$ is called a \emph{power} of the test.}

\subsection{Formalization of Statistical Hypothesis Testing}
\label{sub:reasoning:HT}

Now we define a distributional Kripke model $\M$ with a universe $\calw$ that includes at least two worlds $\wre$ and $\wid$ corresponding to the two datasets we compare: 
\begin{itemize}
\item the real world $\wre$ where we have a dataset sampled from actual experiments (e.g., from a medical treatment whose effectiveness we want to know);
\item the ideal world $\wid$ where we have a dataset that is synthesized from the null hypothesis setting (e.g., the dataset obtained from the placebo effect).
\end{itemize}
Note that $\calw$ may include other worlds corresponding to different possible datasets.

Let $\samp$ be the size of the dataset, and
$x$ be a measurement variable denoting a single data value chosen from the dataset we have.
We assume that each world $w$ has a state $s$ corresponding to each single data value $\sigma_{s}(x)$ in the dataset.
Then $\sigma_{\wre}(x)$ is the empirical distribution (histogram) calculated from the dataset observed in the actual experiments in $\wre$, while $\sigma_{\wid}(x)$ is the distribution calculated from the synthetic dataset in $\wid$.
Then the number of data having a value $v$ in the dataset in a world $w$ is given by $\samp \cdot \sigma_{w}(x)[v]$.

Assume that $\M$ has an accessibility relation $\calr_{c_{\alpha}/\samp}$ that is specific to the sample size $\samp$, the statistical hypothesis test, and the \emph{critical value} $c_{\alpha}$ for a significance level $\alpha$ we use.
For brevity let $\epsan = c_{\alpha} / \samp$.
Intuitively, $(\wre, \wid) \in \Rea$ represents that the hypothesis test cannot distinguish  the actual dataset from the synthetic one.
For instance, when we use Pearson's $\chi^2$-test as the hypothesis test, then $\Rea$ is defined by: 
\[
\Rea \eqdef
\left\{ (w, w') \in \calw\times\calw \mid 
D_{\chi^2}(\sigma_{w}(x) \parallel \sigma_{w'}(x)) \le \epsan
\right\}\!,
\]
where $D_{\chi^2}$ is Pearson's $\chi^2$ divergence (Definition~\ref{def:chi-divergence}).

Observe that 
when the confidence level $1- \alpha$ increases, then $c_{\alpha}$ decreases, hence $\epsan = c_{\alpha} / \samp$ is smaller, i.e., the capability of distinguishing possible worlds is stronger.

Let $\phisyn$ be a formula representing that the dataset is synthesized from the null hypothesis setting (e.g., representing the placebo effect).
Then $\M, \wid \models \phisyn$.
Since each world in $\calw$ corresponds to a different dataset, it holds for any $w'\neq \wid$ that $\M, w' \models \neg \phisyn$.
For instance, $\M, \wre \models \neg \phisyn$, since the actual dataset is used in $\wre$ even when it looks indistinguishable from the synthetic dataset by the hypothesis test.

When the null hypothesis is rejected with a confidence level $1-\alpha$, then $(\wre, \allowbreak \wid) \not\in \Rea$.
Since $\M, w' \models \neg \phisyn$ holds for any $w'\neq \wid$, 
this rejection of the null hypothesis implies:
\[
\M, \wre \models \MKea \neg \phisyn
{,}
\]
which is logically equivalent to $\M, \wre \models \neg \MPea \phisyn$.
This means that with the confidence level $1-\alpha$, we know we are not located in the world $\wid$, hence do not have a synthetic dataset.

On the other hand, when the null hypothesis is not rejected with a confidence level $1-\alpha$, then $(\wre, \wid) \in \Rea$.
Thus we obtain:
\begin{align}\label{eq:null}
\M, \wre \models \MPea \phisyn
{.}
\end{align}
This means that we cannot recognize whether we are located in the world $\wre$ or $\wid$, i.e.,
we are not sure which database we have.
To see this in details, let $\phi'$ be a formula representing that we have a third database (different from those in $\wre$ and $\wid$).
Suppose that another null hypothesis of satisfying $\phi'$ is not rejected with a confidence level $1-\alpha$.
Then we have $\M, \wre \models \MPea \phi'$.
Since each world in $\calw$ corresponds to a different database, we obtain
$\M, \wre \models \MPea \neg \phisyn$, 
which implies 
$\M, \wre \models \neg \MKea \phisyn$.
This represents that, when the null hypothesis is not rejected, 
we are not sure whether the null hypothesis is true or false.

\subsection{Formalization of Statistical Secrecy}
\label{sub:secrecy}
Now let us formalize the coin flipping in Example~\ref{eg:coin} in Section~\ref{sec:running-example} by using \StatEL{} as follows.
Recall that $p(heads) = 0.5$ in $w_0$ and $p(heads) = 0.4$ in $w_1$.
Let $\psi$ be a static formula representing that the coin is a heads.
Then $\M, w_0 \models \PR{0.5} \psi$ and $\M, w_1 \models \PR{0.4} \psi$.
Assume that either $p(heads) = 0.5$ or $p(heads) = 0.4$ holds, i.e., $\M \models \PR{0.5} \psi \vee \PR{0.4} \psi$.

When we have a sufficient number $\samp$ of coin flips (e.g., $\samp=500$), we can distinguish $p(heads) = 0.5$ from $p(heads) = 0.4$ (i.e., $w_0$ from $w_1$) by a hypothesis test.
Hence we learn the probability $p(heads)$ with some confidence level $1 - \alpha$, i.e., 
$\M, w_0 \models \MKea \PR{0.5} \psi$
and
$\M, w_1 \models \MKea \PR{0.4} \psi$.
Therefore we obtain:
\[
\M \models \bigl( \PR{0.5} \psi \rightarrow \MKea \PR{0.5} \psi \bigr) \wedge
\bigl( \PR{0.4} \psi \rightarrow \MKea \PR{0.4} \psi \bigr)
{.}
\]
Note that for a larger sample size $\samp' > \samp$,\, we have $\epsanp = c_{\alpha} / \samp' < c_{\alpha} / \samp = \epsan$, 
hence it follows from the axiom $({\bf{\ge}})$ in Proposition~\ref{prop:axioms:divergence} that: 
\[
\M \models \bigl( \PR{0.5} \psi \rightarrow \MKeap \PR{0.5} \psi \bigr) \wedge
\bigl( \PR{0.4} \psi \rightarrow \MKeap \PR{0.4} \psi \bigr)
{.}
\]
This means that if our knowledge derived from a smaller sample is statistically significant, then we derive the same conclusion from a larger sample.

On the other hand, when we have a very small number $\samp''$ of coin flips, we cannot distinguish $w_0$ from $w_1$.
Then we are not sure about $p(heads)$ with a confidence level $1-\alpha$, i.e., 
$\M, w_0 \models \MPeapp \PR{0.5} \psi$ and
$\M, w_1 \models \MPeapp \PR{0.4} \psi$.
Hence:
\[ 
\M \models (\PR{0.5} \psi \vee \PR{0.4} \psi ) \rightarrow (\MPeapp \PR{0.5} \psi \wedge \MPeapp \PR{0.4} \psi)
{.}
\]
This expresses a secrecy of $p(heads)$.
We generalize this to introduce the following definition of secrecy.
\begin{definition}[$(\alpha,\samp)$-statistical secrecy]\label{def:secrecy}\rm
Let $\Phi$ be a finite set of formulas,
$\alpha \in [0, 1]$ be a significance level,
and $\samp$ be a sample size.
We say that $\Phi$ is \emph{$(\alpha,\samp)$-statistically secret} if we have:
\begin{align*}
\M \models \bigvee_{\phi\in\Phi} \phi \rightarrow \bigwedge_{\phi \in \Phi} \MPea \phi
{.}
\end{align*}
\end{definition}

In the above coin flipping example, $\{ \PR{0.5} \psi,\, \PR{0.4} \psi \}$ is $(\alpha,\samp)$-statistically secret for some significance level $\alpha$ and sample size $\samp$.
Syntactically, $(\alpha,\samp)$-statistical secrecy resembles the notion of \emph{total anonymity}~\cite{Halpern:03:CSFW}, whereas in our definition, the epistemic operator $\MPea$ deals with the statistical significance and $\phi$ is not limited to a formula representing an agent's action.

\section{Modeling Statistical Data Privacy Using \StatEL{}}
\label{sec:formal:DP}
In this section we formalize a notion of statistical data privacy
by using \StatEL{}.

\subsection{Differential Privacy}
\label{sub:DP}

\emph{Differential privacy}~\cite{Dwork:06:ICALP,dwork2014algorithmic} is a popular measure of data privacy guaranteeing that 
by observing a statistics about a database $d$,
we cannot learn whether an individual user's record is included in $d$ or not.

As a toy example, let us assume that the body weight of individuals is sensitive information, and we publish the average weight of all users recorded in a database~$d$.
Then 
we denote by $d'$ the database obtained by adding to $d$ a single record of a new user $u$'s weight.
If we also disclose the average weight of all users in $d'$, then you learn $u$'s weight from the difference between these two averages.

To mitigate such privacy leaks, many studies have proposed \emph{obfuscation mechanisms}, i.e., randomized algorithms that add random noise to the statistics calculated from databases.
In the above example, an obfuscation mechanism receives a database $d$ and outputs a statistics of average weight to which some random noise is added.
Then you cannot learn much information on $u$'s weight from the perturbed statistics of average weight.

The privacy achieved by such obfuscation is often formalized as differential privacy.
Intuitively, an $\varepsilon$-differential privacy mechanism makes every two ``adjacent'' (i.e., close) database $d$ and $d'$ indistinguishable with a degree of $\varepsilon$.

\begin{definition}[Differential privacy] \label{def:DP} \rm
Let $e$ be the base of natural logarithm,
$\varepsilon\ge 0$, 
$\cald$ be the set of all databases, and $\Psi\subseteq\cald\times\cald$ be an adjacency relation between two databases.
A randomized algorithm $\alg: \cald \rightarrow \Dists\calo$ provides 
\emph{$\varepsilon$-differential privacy} 
w.r.t. $\Psi$ if for any $(d, d')\in\Psi$ and any $R\subseteq\calo$,
\[
\Prob[ \alg(d)\in R ] \leq e^{\varepsilon} \,\Prob[ \alg(d')\in R ]
\]
where the probability is taken over the randomness in $\alg$.
\end{definition}

For a smaller $\varepsilon$, 
the protection of differential privacy is stronger.
It is known that differential privacy can be defined using 
the max-divergence $\Dinf$ 
(Definition~\ref{def:max-divergence}) as follows~\cite{dwork2014algorithmic}.

\begin{restatable}{prop}{DP}
\label{prop:DP}
An obfuscation mechanism $\alg: \cald \rightarrow \Dists\calo$ provides 
$\varepsilon$-differential privacy 
w.r.t. $\Psi\subseteq\cald\times\cald$ iff for any $(d, d')\in\Psi$,\,
$\maxdiverge{\alg(d)}{\alg(d')} \le \varepsilon$ and
$\maxdiverge{\alg(d')}{\alg(d)} \le~\varepsilon$.
\end{restatable}

\subsection{Formalization of Differential Privacy}
\label{sub:formalize-DP}

Next we define a distributional Kripke model 
$\M =(\calw, \Reps, (V_s)_{s\in\cals})$
where there is a possible world corresponding to each database in $\cald$.
We assume that each world is a probability distribution of states in each of which an obfuscation mechanism $\alg$ uses a different value of random seed for providing a probabilistically perturbed output.
Let $x$ (resp. $y$) be a measurement variable denoting the input (resp. output) of the obfuscation mechanism $\alg$.
In each world $w$,\, $\sigma_{w}(x)$ is the database that $\alg$ receives as input, and $\sigma_{w}(y)$ is the distribution of statistics that $\alg$ outputs.
Then the set of all databases is denoted by $\cald = \{ \sigma_{w}(x) \mid w\in\calw \}$.

Now we define the accessibility relation $\Reps$
in $\M$ by using 
the max divergence $\Dinf$ 
as follows\footnote{Since the relation $\Reps$
is symmetric, the symmetry axiom {\rm ({\bf B})} also holds. }:
\[
\Reps
\eqdef 
\left\{ (w, w') \in \calw\times\calw \mid 
\maxdiverge{\sigma_{w}(y)\!}{\!\sigma_{w'}(y)} \le \varepsilon
\mbox{, }
\maxdiverge{\sigma_{w'}(y)\!}{\!\sigma_{w}(y)} \le \varepsilon
\right\}
.
\]
Intuitively, 
$(w, w') \in \Reps$ 
represents that, when we observe an output $y$ of the obfuscation mechanism $\alg$, we do not know which of the two worlds $w$ and $w'$ we are located at.
Hence we do not see which of the two databases $\sigma_{w}(x)$ and $\sigma_{w'}(x)$ was the input to $\alg$.

For each $d\in\cald$, let $\phi_d$ be a formula representing that we have a database $d$.
Then the 
$\varepsilon$-differential privacy 
of $\alg$ w.r.t. an adjacency relation $\Psi$ is expressed~as:
\[
\M \models 
\bigwedge_{d\in \cald} 
\Bigl( \phi_d \rightarrow
\bigwedge_{d'\in \Psi(d)} 
\MPe \phi_{d'}
\Bigr)
{.}
\]

Note that the privacy of user attributes defined as distribution privacy~\cite{Kawamoto:19:ESORICS} can also be expressed using \StatEL{}, since it is defined as the differential privacy w.r.t. a relation between the probability distributions that represent user attributes.
We will elaborate on this in future work.

\section{Related Work}
\label{sec:related}
In this section, we overview related work, including the integration of logical and statistical techniques, epistemic logic, and logical formalization of privacy.

\paragraph{Integration of logical and statistical techniques.}
There have been various studies on integrating logical and statistical techniques in software engineering.
Notable examples are \emph{probabilistic programming}~\cite{Gordon:14:FOSE}, which has sampling from distributions and conditioning by observations, and 
\emph{statistical model checking}~\cite{Sen:04:CAV,Younes:05:PhD,Legay:10:RV}, which checks the satisfiability of logical formulas by simulations and statistical hypothesis tests.
In research of privacy, a few papers (e.g.,~\cite{Biondi:18:FAOC}) present hybrid methods combining symbolic and statistical analyses to quantify privacy leaks.
In future work, our logic may be used to define specifications of these techniques and characterize their properties.

\paragraph{Non-determinism and probability in Kripke models.}
Although many epistemic models have been proposed~\cite{Fagin:95:book,Halpern:03:book,Halpern:03:CSFW},
they often assume that each possible world is a single deterministic state.
To formalize the behaviours of stochastic systems in their model, they assume that every world is assigned a probability (e.g.,~\cite{Kooi:03:JOLLI}), which means the non-determinism needs to be resolved in advance.

However, not only probability but also non-deterministic inputs are essential to reason about security and many applications in statistics.
In the context of security, we usually do not have a prior knowledge of the probability distribution of adversarial inputs.
Also in the statistical hypothesis testing (Section~\ref{sub:HT}), we do not assume the prior probabilities of the null/alternative hypotheses.
The notion of differential privacy (Definition~\ref{def:DP}) is also independent of the prior distribution on the databases.
Therefore, unlike ours, the Kripke models in previous work cannot be used for the purpose of formalizing such statistical knowledge.

\paragraph{Kripke model for some aspects of statistics.}
The \emph{random worlds model}~\cite{Halpern:03:book} is an epistemic model that tries to formalize some aspects of statistics.
In that model, they assume that each possible world has an identical probability at the initial time, although this causes problems as mentioned in Chapter 10 of~\cite{Halpern:03:book}.
Unlike our distributional model, their model employs neither distributions of states nor statistical significance.
They assume only finite intervals of errors, and analyze only the ideal situation that corresponds to an infinite sample size.
Therefore, the random worlds model cannot formalize statistical knowledge in our sense.

In research of philosophical logic, \cite{Lewis:80:subjectivist,Bana:17:EPSP} formalize the idea that when a random value has various possible probability distributions, those distributions should be represented on different possible worlds.
Unlike our work, however, they do not model statistical significance or explore accessibility relations.

Independently of our work, French et al.~\cite{French:19:AAAMAS} propose a probability model for a dynamic epistemic logic where each world is associated with a (subjective) probability distribution \emph{over the universe} and may have a different probability for a propositional variable to be true.
This is different from our distributional Kripke model in that their model does not associate each world with a probability distribution of observable variables, hence deals with neither non-deterministic inputs, divergence-based accessibility relations, nor statistical significance.

\paragraph{Epistemic logic for privacy properties.}
Epistemic logic has been used to formalize and reason about privacy properties, including anonymity~\cite{Syverson:99:FM,Halpern:03:CSFW,Meyden:04:CSFW,Garcia:05:FMSE,Kawamoto:07:JSIAM,Eijck:07:ENTCS,Baskar:07:TARK,Chadha:09:Forte}, role-interchangeability~\cite{Mano:10:JLC}, receipt-freeness of electronic voting protocols~\cite{Jonker:06:WOTE,Baskar:07:TARK}, and its extension called coercion-resistance~\cite{Kuesters:09:SP}.
Unlike our formalization in Section~\ref{sec:formal:DP}, however, these do not regard possible worlds as probability distributions and cannot formalize privacy properties with a statistical significance.

\paragraph{Logical approaches to differential privacy.}
There have been studies that formalize differential privacy using logics, such as Hoare logic~\cite{Barthe:14:CSF} and \mbox{HyperPCTL}~\cite{Abraham:18:QEST}.
Compared to \StatEL{}, these formalizations need to explicitly describe inequalities of probabilities without much abstraction, hence the formulas are more complicated.
In addition, none of them formalizes the situation with finite sample sizes or statistical significance.

\section{Conclusion}
\label{sec:conclude}
We introduced statistical epistemic logic (\StatEL{}) to describe statistical knowledge,
and showed its stochastic semantics based on the distributional Kripke model.
By using \StatEL{} we introduced $(\alpha,\samp)$-statistical secrecy with a significance level $\alpha$ and a sample size $\samp$, and showed that \StatEL{} is useful to formalize hypothesis testing and differential privacy in a simple way.
As shown in~\cite{Kawamoto:19:SEFM}, \StatEL{} can also express certain properties of statistical machine learning.

In our ongoing work, we extend \StatEL{} to deal with the security of cryptography based on computational complexity theory.
As for future work, we will extend this logic with temporal modality and give its axiomatization.
Our future work includes an extension of \StatEL{} to formalize the quantitative notions of anonymity~\cite{Chatzikokolakis:08:IC} and asymptotic anonymity~\cite{Kawamoto:18:ISITA}.
We are also interested in clarifying the relationships between our distributional Kripke model and the main stream probabilistic epistemic logic assigning probabilities to worlds.
Furthermore, we plan to develop statistical epistemic logic for process calculi in an analogous way to~\cite{Chadha:09:Forte,Hughes:04:JCS,Dechesne:07:LPAR,Chatzikokolakis:12:TCL}, 
and to investigate the relationships between statistical epistemic logic and bisimulation metrics analogously to~\cite{Chatzikokolakis:14:concur}.

\section*{Acknowledgments}
I would like to thank the reviewers for their helpful and insightful comments.
I am also grateful to Ken Mano, Gergei Bana, and Ryuta Arisaka for their useful comments on preliminary manuscripts.

 \bibliographystyle{splncs04}
 \bibliography{short}

\arxiv{
\appendix
\section{Properties of Probability Quantification}
\label{sec:proofs:quantifier}
In this section we present the proofs for properties of probability quantification.

\AxiomsNotPIQ*
\begin{proof}
We show the first claim as follows.
By the definition of semantics,
$\M, w \models \neg \PR{I} \psi$ is logically equivalent to 
$\Pr\!\left[ s \randassign w :~ s \models \psi \right] \not\in I$, 
which is equivalent to
$\Pr\!\left[ s \randassign w :~ s \models \psi \right] \in~I^c$,
namely, $\M, w \models \PR{I^c} \psi$.

Next we show the second claim as follows.
By the definition of semantics,
$\M, w \models \PR{I} \neg \psi$ is logically equivalent to 
$1 - \Pr\!\left[ s \randassign w :~ s \models \psi \right] \in I$, i.e.,
$\Pr\!\left[ s \randassign w :~ s \models \psi \right] \in \ov{I}$.
This is equivalent to $\M, w \models \PR{\ov{I}} \psi$.
\qed
\end{proof}

\section{Properties of the Epistemic Operators}
\label{sec:proofs:epistemic}
In this section we present properties of our epistemic operators and their proofs.

\AxiomsMinimal*
\begin{proof}
We first show ({\bf N}) necessitation rule as follows.
Assume that $\M \models \phi$.
Then for any world $w$ in $\M$, we have $\M, w \models \phi$.
Hence $\M, w \models \MKa \phi$.
Therefore the necessitation rule holds.

Next we show ({\bf K}) distribution axiom as follows.
Let $w$ be a possible world in $\M$.
Assume that $\M, w \models \MKa (\phi_0 \rightarrow \phi_1)$, and that $\M, w \models \MKa \phi_0$.
Let $w'$ be any world such that $(w, w')\in\Raeps$.
Then we have $\M, w' \models \phi_0 \rightarrow \phi_1$ and $\M, w' \models \phi_0$, hence $\M, w' \models \phi_1$.
Thus we have $\M, w \models \MKa \phi_1$.
Therefore we obtain $\M \models \MKa (\phi_0 \rightarrow \phi_1) \rightarrow  (\MKa \phi_0 \rightarrow \MKa \phi_1)$.
\qed
\end{proof}

\AxiomsDivergence*
\begin{proof}
Let $w$ be a possible world in $\M$.

We first show ({\bf T}) reflexivity as follows.
Assume that $\M, w \models \MKaz \phi$.
By $(w, w)\in\Raz$, we have $\M, w \models \phi$.
Therefore, we obtain $\M \models \MKae \phi \rightarrow \phi$.

Next we show ($\bf{\ge}$) comparison of observability as follows.
Assume that $\M, w \models \MKae \phi$.
Let $w'$ be any world such that $(w, w')\in\Raeps$.
Then $\M, w' \models \phi$.
By $\varepsilon' \le \varepsilon$ and the definition of $\Raeps$, we have $\calr_{a,\varepsilon'}\subseteq\Raeps$, hence $(w, w')\in\calr_{a,\varepsilon'}$.
Then $\M, w \models {\mathop{\mathsf{K}_{a,\varepsilon'}}} \phi$.
Therefore we obtain 
$\M \models {\mathop{\mathsf{K}_{a,\varepsilon}}} \phi \rightarrow {\mathop{\mathsf{K}_{a,\varepsilon'}}} \phi$.

Finally, we show ({\bf B}) symmetry when $\Raeps$ is symmetric.
Assume that $\M, w \models \phi$.
Let $w'$ be any world such that $(w, w')\in\Raeps$.
Since $\Raeps$ is symmetric, we have $(w', w)\in\Raeps$.
By $\M, w \models \phi$, we obtain $\M, w' \models \MPae \phi$.
Hence $\M, w \models \MKae \MPae \phi$.
Therefore we obtain $\M \models \phi \rightarrow \MKae \MPae \phi$.
\qed
\end{proof}

\AxiomsMetric*
\begin{proof}
Since a metric satisfies the definition of a divergence (in Section~\ref{sec:preliminaries}), a metric-based accessibility relation is also a divergence-based accessibility relation.
Therefore we obtain {\rm ({\bf T})} reflexivity and {\rm ({\bf B})} symmetry from Proposition~\ref{prop:axioms:divergence}.

Next we show ({\bf 4q}) quantitative transitivity as follows.
Let $w$ be a possible world in $\M$.
Assume that $\M, w \models {\mathop{\mathsf{K}_{a,\varepsilon+\varepsilon'}}} \phi$.
Let $w'$ be any world such that $(w, w')\in\Raeps$, and
$w''$ be any world such that $(w', w'')\in\calr_{a,\varepsilon'}$.
By definition, we have 
$\diverge{\sigma_{w}(x)}{\sigma_{w'}(x)} \le \varepsilon$ and
$\diverge{\sigma_{w'}(x)}{\sigma_{w''}(x)} \le \varepsilon'$.
By the subadditivity of the divergence $D$, we have 
$\diverge{\sigma_{w}(x)}{\sigma_{w''}(x)} \le 
\diverge{\sigma_{w}(x)}{\sigma_{w'}(x)} + \diverge{\sigma_{w'}(x)}{\sigma_{w''}(x)} \le
\varepsilon + \varepsilon'$,
hence $(w, w'')\in\calr_{a,\varepsilon+\varepsilon'}$.
Then it follows from $\M, w \models {\mathop{\mathsf{K}_{a,\varepsilon+\varepsilon'}}} \phi$ that $\M, w'' \models \phi$.
By the definition of $w''$, we obtain 
$\M, w \models \MKae {\mathop{\mathsf{K}_{a,\varepsilon'}}} \phi$.
Therefore we have
$\M \models {\mathop{\mathsf{K}_{a,\varepsilon+\varepsilon'}}} \phi \rightarrow \MKae {\mathop{\mathsf{K}_{a,\varepsilon'}}} \phi$.

We next show ({\bf 5q}) relaxed Euclidean as follows.
Let $w$ be a possible world in $\M$.
Assume that $\M, w \models {\mathop{\mathsf{P}_{a,\varepsilon}}} \phi$.
Then there exists a world $w'$ such that $\M, w' \models \phi$ and $(w, w') \in \Raeps$.
Let $w''$ be any world such that $(w, w'')\in\calr_{a,\varepsilon'}$.
Since $\calr_{a,\varepsilon'}$ is a metric-based accessibility relation, it is symmetric, hence $(w'', w) \in\calr_{a,\varepsilon'}$.
Then by $(w, w') \in \Raeps$, we obtain $(w'', w')\in\calr_{a,\varepsilon+\varepsilon'}$.
Hence
$\M, w \models {\mathop{\mathsf{K}_{a,\varepsilon'}}} {\mathop{\mathsf{P}_{a,\varepsilon+\varepsilon'}}} \phi$.
Therefore we obtain 
$\M \models {\mathop{\mathsf{P}_{a,\varepsilon}}} \phi \rightarrow {\mathop{\mathsf{K}_{a,\varepsilon'}}} {\mathop{\mathsf{P}_{a,\varepsilon+\varepsilon'}}} \phi$.

Finally, for $\varepsilon = \varepsilon' = 0$, ({\bf 4}) transitivity and ({\bf 5}) Euclidean respectively follow from ({\bf 4q}) quantitative transitivity and ({\bf 5q}) relaxed Euclidean.
\qed
\end{proof}

}

\end{document}